\documentclass{thesis}
\usepackage{fancyvrb}
\usepackage{fullpage}
\usepackage[utf8]{inputenc}
\usepackage[T1]{fontenc}
\usepackage{graphics}
\usepackage{varwidth}

\usepackage{amsmath}
\usepackage{amssymb}
\usepackage{placeins}
\usepackage{amsthm}
\usepackage[noend]{algpseudocode}
\usepackage{tikz}
\usepackage{comment}
\usetikzlibrary{arrows,decorations.pathreplacing}
\newcommand{\atomic}[4]{\node(lab#1) at (#3,-#4+0.5) {#2};\draw (#3,-#4-0.1) to (lab#1);\node[inner sep=0pt](#1) at (#3,-#4) {};}
\algblockdefx{Struct}{EndStruct}[1]{\textbf{struct} \textsc{#1}}{}
\algnewcommand{\algorithmicgoto}{\textbf{go to}}
\algnewcommand{\Goto}{\textbf{goto} }
\algnewcommand{\Label}{\State\unskip}
\def\none{\textbf{none} }
\def\T{\ensuremath{\operatorname{\mathcal{T}}}\text{ }}
\def\int{\ensuremath{\operatorname{\textbf{int}}}}
\algnotext{EndStruct}
\newtheorem{lemma}{Lemma}
\newtheorem{observation}{Observation}
\newtheorem{theorem}{Theorem}
\newtheorem{example}{Example}
\newtheorem{definition}{Definition}

\newcommand{\fn}[1]{\textsc{#1}}
\newcommand{\var}[2]{\textbf{var }#1 : #2}
\newcommand{\arrayspec}[1]{\text{array}[#1]\text{ of }}
\usepackage{float}
\newfloat{Listing}{t}{lop}
\author{Robert Obryk}
\title{Write-and-f-array: implementation and an application}
\nralbumu{1040504}
\kierunek{Computer Science - IT Analyst}
\opiekun{dr Grzegorz Herman}
\date{September 2013}
\VerbatimFootnotes
\begin{document}

\maketitle
\begin{abstract}
	We introduce a new shared memory object: the write-and-f-array, provide its wait-free implementation and use it to construct an improved wait-free implementation of the fetch-and-add object. The write-and-f-array generalizes single-writer write-and-snapshot\cite{write-and-snap} object in a similar way that the f-array\cite{f-array} generalizes the multi-writer snapshot object. More specifically, a write-and-f-array is parameterized by an associative operator $f$ and is conceptually an array with two atomic operations:

\begin{itemize}
\item write-and-f modifies a single array's element and returns the result of applying $f$ to all the elements,
\item read returns the result of applying $f$ to all the array's elements.
\end{itemize}

We provide a wait-free implementation of an $N$-element write-and-f-array with $O(N \log N)$ memory complexity, $O(\log^3 N)$ step complexity of the write-and-f operation and $O(1)$ step complexity of the read operation. The implementation uses CAS objects and requires their size to be $\Omega(\log M)$, where $M$ is the total number of write-and-f operations executed. We also show, how it can be modified to achieve $O(\log^2 N)$ step complexity of write-and-f, while increasing the memory complexity to $O(N \log^2 N)$.

The write-and-f-array can be applied to create a fetch-and-add object for $P$ processes with $O(P \log P)$ memory complexity and $O(\log^3 P)$ step complexity of the fetch-and-add operation. This is the first implementation of fetch-and-add with polylogarithmic step complexity and subquadratic memory complexity that can be implemented without CAS or LL/SC objects of unrealistic size\cite{ellen-fai}.


\end{abstract}
\begin{abstract}
Wprowadzamy nowy obiekt współbieżny: write-and-f-array, podajemy jego implementację wait-free i używamy jej aby skonstruować ulepszoną implementację wait-free obiektu fetch-and-add. Write-and-f-array uogólnia write-and-snapshot dla
jednego pisarza w podobny sposób co f-array uogólnia snapshot. Dokładniej, write-and-f-array jest sparametryzowana przez łączny operator $f$ i koncepcyjnie jest tablicą z dwoma operacjami atomowymi:

\begin{itemize}
\item write-and-f, która modyfikuje jeden element tablicy i zwraca wynik aplikacji $f$ do wszystkich jej elementów,
\item read, która zwraca wynik aplikacji $f$ do wszystkich elementów tablicy.
\end{itemize}

Podajemy implementację wait-free $N$-elementowej write-and-f-array ze złożonością pamięciową $O(N \log N)$, złożonością krokową write-and-f $O(\log^3 N)$ i stałą złożonością operacji read. Implementacja ta
używa obiektów CAS o rozmiarze $\Omega(\log M)$, gdzie $M$ jest całkowitą liczbą wykonanych operacji write-and-f. Pokazujemy też modyfikację tej implementacji, która zmniejsza złożoność krokową write-and-f do $O(\log^2 N)$,
jednocześnie zwiększając złożoność pamięciową do $O(N \log^2 N)$.

Write-and-f-array znajduje zastosowanie w konstrukcji obiektu fetch-and-add dla $P$ procesorów ze złożonością pamięciową $O(P \log P)$ i złożonością krokową operacji $O(\log^3 P)$. Jest to pierwsza implementacja fetch-and-add
z polilogarytmiczną złożonością krokową operacji i podkwadratową złożonością pamięciową, która nie wymaga obiektów CAS lub LL/SC o nierealistycznie dużym rozmiarze.

Słowa kluczowe: współbieżność, wait-free, snapshot, fetch-and-add
\end{abstract}
\clearpage
\tableofcontents
\clearpage
\chapter{Introduction}

The advent of popular manycore systems has made concurrent programming for shared memory systems an important
topic. With increasing parallelism available in a single shared-memory system, performance of techniques
used to communicate between concurrently executing threads, or to access memory common to many such threads,
is becoming one of the most significant components of the performance of an application running on such a system.
In many cases of communication-heavy tasks, simple coarse-grained locking is not enough to yield good performance,
so programmers need to resort to lock-less communication and concurrent data structures. This work provides a new
data structure with an implementation that can be used concurrently and doesn't use locks, and uses it to create
an implementation of the fetch-and-add object (a kind of counter) with improved memory usage.

\section{Concurrent Objects}
A useful abstraction in a single-threaded system is an object: an entity with a set of methods one can invoke,
and a specification of semantics of these methods. Naturally, a program can use multiple objects and invoke
their methods in any order.

We want to use a similar abstraction to model interaction between threads in a multi-threaded system.
Consider a set of independent, concurrently running threads (running possibly
distinct code) that can only communicate by calling methods on some objects. We place no assumptions on
the speed of execution and delays of the threads -- they may wait arbitrarily long before a method call.
We also assume that method calls on the objects complete instanteously, and that no two of them happen at the same time.
This allows us to define semantics of the objects in the same fashion we define them in the single-threaded case:
methods of every object are invoked in a known order, so single-threaded semantics suffice to determine the object's behaviour\footnote{The whole program may still be nondeterministic, due to nondeterministic scheduling.}.
We will call such objects \emph{concurrent objects} to emphasize that they may be accessed concurrently.




An object we will predominantly use in this work is a CAS (Compare-And-Swap) register.
We will specify its semantics by providing pseudocode which correctly implements this object in
a single-threaded program (this will be our method of choice of providing object's semantics):

\begin{algorithmic}[1]
\State\var{$x$}{\T}
	
\Function{rd}{}
	\State\Return $x$
\EndFunction

\Function{wr}{$y$}
	\State $x \gets y$
\EndFunction
\pagebreak
\Function{CAS}{$o, n$}
\If{$x = o$}
	\State $x \gets n$
	\State\Return true
\Else
	\State\Return false
\EndIf
\EndFunction
\end{algorithmic}

Intuitively, such a register holds a value that can be read (by \fn{rd}), modified unconditionally (by \fn{wr})
and modified conditionally (by \fn{CAS}). It's interesting due to its universality properties\cite{cas-univ} 
and because it is
commonly seen in real-world hardware. 

\begin{example}\label{ex-simple}
Let's consider two threads, executing following pseudocodes (X and Y are two CAS objects, with initial value 0):

\begin{enumerate}
\item\begin{algorithmic}[1]
\State X.\fn{wr}(1)
\State Y.\fn{rd}()
\end{algorithmic}

\item\begin{algorithmic}[1]
\State Y.\fn{wr}(1)
\State X.\fn{rd}()
\end{algorithmic}
\end{enumerate}

Figure~\ref{fig-ex-simple} presents a possible execution of such two threads. Note that in any correct execution,
at least one of the \fn{rd}() calls returns 1.
\end{example}

\begin{figure}[!h]
\begin{tikzpicture}
	\atomic{x_wr_1}{X.\fn{wr}(1)}{0.5}{0}
	\atomic{y_rd_1}{Y.\fn{rd}()=0}{2.5}{0}
	\atomic{y_wr_1}{Y.\fn{wr}(1)}{3.5}{1}
	\atomic{x_rd_1}{X.\fn{rd}()=1}{5.5}{1}
	\draw (0,0) to (6,0);
	\draw (0,-1) to (6,-1);

	\atomic{x_wr_2}{X.\fn{wr}(1)}{8.0}{0}
	\atomic{y_rd_2}{Y.\fn{rd}()=1}{10.0}{0}
	\atomic{y_wr_2}{Y.\fn{wr}(1)}{8.5}{1}
	\atomic{x_rd_2}{X.\fn{rd}()=1}{10.5}{1}
	\draw (7.5,0) to (11,0);
	\draw (7.5,-1) to (11,-1);


\end{tikzpicture}

\caption{Two possible interleaved executions of threads from Example~\ref{ex-simple}.}
\label{fig-ex-simple}
\end{figure}
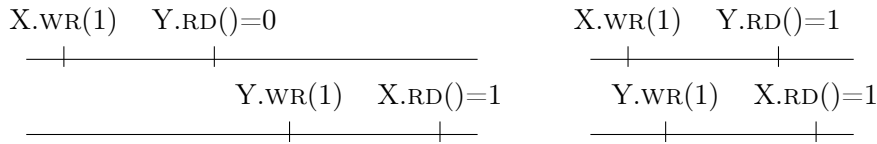

\section{Serialization}
In many situations it would be convenient to have more complicated concurrent objects: counters that can be incremented,
queues that can be pushed to and popped from, etc. Alas, real-world hardware usually doesn't implement such objects
directly. We will thus implement them using CAS registers or other simpler objects: we will substitute a set of simpler objects for every such object and for every method call on this object, execute a procedure that implements it instead.
However, these prodecures might run concurrently, so we can't blindly use a single-threaded implementation of the object.
Obviously, we will want such substitution to preserve semantics of the program in which it was substituted. The following condition immediately implies
this\footnote{For a detailed discussion of subtleties related to this definition, refer to chapter 3.3-3.6 of \cite{TheArt}}:

\begin{definition}
An implementation of a concurrent object \emph{serializes} iff for every execution of a program with the implementation
substituted for the object:
For every call to object's method that completes, we can define a \emph{serialization point} in its execution
interval (between its first and last constituent operation), such that all these calls would return the same values if they happened instaneously at their serialization
points.
\end{definition}

An observant reader might note that if all procedures in an implementation never terminate, the definition is
vacously true. Indeed, we will present some terminations guarantees separately, but first let us consider an
example implementation of a simple structure.





\begin{example}
Let's consider an increase-only counter: a shared object with the following semantics:

\begin{algorithmic}[1]
\State\var{$x$}{\int}

\Function{inc}{}
	\State$x \gets x + 1$
\EndFunction
\Function{read}{}
	\State\Return$x$
\EndFunction
\end{algorithmic}

We will see that the following implementation serializes to this shared object:

\begin{algorithmic}[1]
\State\var{$V$}{\int} \Comment{A CAS object holding an \int}

\Function{inc}{}
	\Repeat
	\State$o \gets V.\fn{rd}()$ \label{trivial-faa-rd}
	\Until{$V.\fn{CAS}(o, o+1)$}
\EndFunction

\Function{read}{}
	\State\Return $V.\fn{rd}()$
\EndFunction
\end{algorithmic}

Note that every call to \fn{inc} executes exactly one successful call to $V.\fn{CAS}$. Let us claim that the serialization
point of \fn{inc} is that successful call to \fn{CAS}, and that the serialization point of \fn{read} is the call to
$V.\fn{rd}$. We can easily see that the number returned by \fn{rd} is exactly the number of \fn{inc}s that serialized
before that \fn{read}.
\end{example}

We can also see that this implementation satisfies some global progress property -- when the CAS in \fn{inc} fails,
another \fn{inc} must have suceeded very recently. In fact, it must have succeeded between the most recent execution
of line~\ref{trivial-faa-rd} and the failed CAS. Still, a single \fn{inc} operation can fail to terminate.

\section{Wait-freedom}

Consider an alternative implementation of an increment-only counter. In the following pseudocode, $P$ denotes
the number of threads in the system.

\begin{algorithmic}[1]
	\State\var{$T$}{\arrayspec{P}$\int$} \Comment{T is an array of $P$ CAS objects of type \int}
\Function{inc$_p$}{}
	\State $x \gets T[p].\fn{rd}()$
	\State $T[p].\fn{wr}(x + 1)$
\EndFunction
\Function{read}{}
	\State $R \gets 0$
	\For{$i \gets 1 \mbox{ to } P$}
	\State $R \gets R + T[i].\fn{rd}$
	\EndFor
	\State\Return $R$
\EndFunction
\end{algorithmic}

Here we differentiate between threads: the index $p$ in \fn{inc} corresponds to the current thread's ID from range
$\{0, \ldots, P-1\}$. Let us first see that this implementation actually serializes to the increment-only counter.
Obviously, we need to choose the call to WRITE as the serialization point of \fn{inc}. From the monotonicity
of entries of T we can infer that \fn{read} returns a value between (inclusive) the number of calls to \fn{inc} that have
serialized before \fn{read} has started and the number of calls serialized before it has finished.
Thus if a call to \fn{read} returns $x$, there is a point during its
execution when exactly $x$ \fn{inc}s had been serialized. We can choose this point to be the serialization point of
\fn{read}. Note that it need not correspond to an action executed by the implementation of \fn{read}.

We can see that this implementation satisfies a very strong termination condition -- the number of steps every
procedure takes is bounded and the bound depends only on the number of threads (and not on the behaviour of
the scheduler). We will call such implementations \emph{wait-free}. For a discussion of weaker
termination guarantees, see chapter 3.7 of \cite{TheArt}.

Every object that has a single-threaded specification has a wait-free implementation\cite{cas-univ} using
only CAS registers.
However, all currently-known such constructions yield objects with step complexities inflated at least by a factor
of the number of threads.
Thus, creating faster wait-free implementations of data structures is interesting, also from a practical
point of view.

\section{Snapshots}

The object implemented in this work is a generalization of a popular building block for multithreaded objects
-- an atomic snapshot object\cite{snapshot-idea}. Its semantics can be defined by the following single-threaded pseudocode:

\begin{algorithmic}[1]
	\State\var{$T$}{\arrayspec{N}$\T$}
\Function{update$_i$}{$x$}
	\State $T[i] \gets x$
\EndFunction
\Function{scan}{}
	\For{$i \gets 0 \mbox{ to } N$}
		\State $R[i] \gets T[i]$
	\EndFor
	\State\Return $R$
\EndFunction
\end{algorithmic}

Intuitively, this object represents an array that can be modified piecewise and read all at once.
There are known wait-free implementations of an $N$-element atomic snapshot with $O(1)$ \fn{update} and $O(N)$ \fn{scan} step complexities\cite{snapshot-impl}. These complexities are obviously
optimal.

One can try to extend this structure in many ways. Two of them are particularly interesting for us.
First, we can merge the update and scan operations
into an operation that does an update and a scan\footnote{Obviously, one can call these operations sequentially
on a normal snapshot object. However, another update operation could then happen between the update and the scan.
This cannot happen if the update and scan calls are merged into a single atomic operation.}.
This object, which we will call write-and-snapshot\cite{write-and-snap}, can be defined by the following single-threaded specification:

\begin{algorithmic}[1]
	\State\var{$T$}{\arrayspec{N}$\T$}
\Function{update-and-scan$_i$}{$x$}
	\State $T[i] \gets x$
	\For{$i \gets 0 \mbox{ to } N$}
		\State $R[i] \gets T[i]$
	\EndFor
	\State\Return $R$
\EndFunction
\end{algorithmic}

For the second extension, let us note that in many cases one doesn't need the \fn{scan} operation to return the whole array, but rather some sort of its digest. The f-arrays\cite{f-array} allow one to do precisely that for digests expressible as
a result of applying an associative operation to all the array's elements (for example, if we need to retrieve the sum of them). If we call the operator $f$, the following is a specification of an f-array:

\begin{algorithmic}[1]
	\State\var{$T$}{\arrayspec{N}$\T$}
\Function{update$_i$}{$x$}
	\State $T[i] \gets x$
\EndFunction
\Function{scan}{}
	\State\Return $f(T[0], T[1], \ldots, T[N-1])$
\EndFunction
\end{algorithmic}

There is an implementation of f-array with step complexities of \fn{update} and \fn{scan} being respectively $O(\log N)$ and $O(1)$ and the memory complexity
being linear in $N$.

This work introduces a write-and-f-array: an object that combines these two modifications. Its semantics are defined by the single-threaded specification below:

\begin{algorithmic}[1]
	\State\var{$T$}{\arrayspec{N}$\T$}
	\Function{update-and-$f_i$}{$x$}
		\State $T[i] \gets x$
		\State\Return $f(T[0], T[1], \ldots, T[N-1])$
	\EndFunction
\end{algorithmic}

The main result of this work is an implementation of a single-writer
write-and-f-array: that is, one that disallows concurrent modifications to the same array element. The implementation
uses $O(N\log{}N)$ memory and the step complexity of the operation is $O(\log^3N)$. 

We then use this object to implement a fetch-and-add object. Fetch-and-add is a generalization of the increment-only counter from previous chapters. Its semantics are specified by the listing below:

\begin{algorithmic}[1]
\State\var{V}{\int}
\Function{fetch-and-add}{$x$}
	\State $r \gets V$
	\State $V \gets V + x$
	\State\Return $r$
\EndFunction
\end{algorithmic}

Fetch-and-add object can be used to produce unique identifiers, implement
mutual exclusion, barrier synchronization\cite{faa-sync}, or work queues\cite{faa-queue}. Its known wait-free shared memory implementations for $P$ processes have $\Omega(P^2)$ memory complexity and $O(\log^2 P)$ step complexity
\cite{ellen-fai}\cite{closed-object}. They also need to employ complicated memory management techniques to bound their memory use. Our implementation reduces the memory complexity to $O(P\log P)$ while maintaining polylogarithmic step complexity.


\chapter{The History Object}

We will first implement a helper object -- the history object. Intuitively, it contains a versioned memory cell and allows retrieval of past $N$ values of the cell. The cell holds objects of type \T. The semantics of this object are precisely specified by the following single-threaded implementation:
\begin{algorithmic}[1]
	\State\var{$H$}{\arrayspec{\ldots}$\T$} \Comment{An unbounded array of values}
	\State\var{$V$}{\int} \Comment{Current version}
	\Function{get-current}{}
		\State \Return $V$, $H[V]$
	\EndFunction
	\Function{get}{$v$}
	\If{$v \leq V - N \textbf{ or } v > V$}
			\State \Return \none
		\Else
		\State \Return $H[v]$
		\EndIf
	\EndFunction
	\Function{publish}{$v$, $T$}
		\If{$v = V + 1$}
			\State $H[v] \gets T$
			\State $V \gets v$
			\State \Return true
		\Else
			\State \Return false
		\EndIf
	\EndFunction
\end{algorithmic}

Our wait-free implementation will impose an additional constraint on its use: every call to \fn{publish} is parameterized by an integer in range $[0,P)$ and executions of \fn{publish$_p$} for the same $p$ can not run simultaneously\footnote{An obviously correct choice for the parameter is a unique thread ID.}.
	Both $N$ and $P$ must be known when the object is created.

Our implementation is presented below:

\begin{algorithmic}[1]
	\State\var{$S$}{$\textbf{pair}\left<v: \int, p: \int\right>$} \Comment{Version number and $p$ of the most recently published value}
	\State\var{$H$}{\arrayspec{N}$\textbf{pair}\left<v: \int, T: \T\right>$} \Comment{A circular buffer of recent $N$ values}
	\State\var{$L$}{\arrayspec{P}$\textbf{pair}\left<v: \int, T: \T\right>$} \Comment{Temporary storage for values being published}
	\Function{get-current}{}
		\State $s \gets S.\fn{rd}()$ \Comment{Possible serialization point}
		\State \fn{help}()
		\State \Return $H[s.v \mod N].\fn{rd}()$
	\EndFunction
	\Function{help}{} \Comment{Updates $H[]$ with value from $L[]$, if required}
		\State $s \gets S.\fn{rd}()$
		\State $l \gets L[s.p].\fn{rd}()$
		\State $h \gets H[s.v \mod N].\fn{rd}()$
		\If{$l.v = s.v \textbf{ and } h.v < s.v$}
			\State $H[s.v \mod N].\fn{CAS}(h, l)$ \label{hist-help-cas}
		\EndIf
	\EndFunction
	\pagebreak
	\Function{get}{v}
		\State $s \gets S.\fn{rd}()$ \label{hist-get-notyet-sp} \Comment{Possible serialization point}
		\If{$s.v < v$}
			\State \Return \none \label{hist-get-notyet} \Comment{The requested version hasn't been published yet}
		\EndIf
		\State \fn{help}()
		\State $h \gets H[v \mod N].\fn{rd}()$
		\If{$h.v = v$}
			\State \Return $h.T$
		\Else
			\State \Return \none \label{hist-get-tooold}
		\EndIf
	\EndFunction
	\Function{publish$_p$}{$v, T$}
		\State $s \gets S.\fn{rd}()$ \label{hist-pub-earlier-sp} \Comment{Possible serialization point}
		\If{$v \neq s.v + 1$}
			\State \Return \textbf{false} \label{hist-pub-earlier-exit}
		\EndIf
		\State \fn{help}()
		\State $L[p].\fn{wr}((v, T))$ \label{hist-pub-latest}
		\If{$S.\fn{CAS}(s, (v, p)$) \textbf{failed}} \Comment{Possible serialization point}\label{hist-pub-later-sp}
			\State \Return \textbf{false}
		\EndIf
		\State \Return \textbf{true}
	\EndFunction
\end{algorithmic}

As the comments indicate, $L[p]$ is used to temporarily hold the value being published by \fn{publish$_p$}. The array $H$ is used to hold, roughly,
the $N$ most recently published values (the most recently published value might be absent; exact semantics will be given by the lemmas below) together with their version numbers. The field $S$ holds the particulars of the most recently published value; the successful publish calls will serialize at the moment $S$ changes. Our implementation contains an additional function \fn{help}. It is used internally to write the most recently published value to the array $H$, if it isn't stored there already.

We will now prove some properties related to exact semantics of $H$ and $L$ and then use them to prove that our implementation serializes to the specification.


We posit that the serialization point of \fn{publish} is in~line~\ref{hist-pub-later-sp}, unless that line isn't reached. In that case (ie. when the call returns false in~line~\ref{hist-pub-earlier-exit})
the serialization point is in~line~\ref{hist-pub-earlier-sp}.


\begin{lemma}\label{hist-cv-latest}
If $S = (v, p)$ at time $t$, and at some later point in time $L[p] = (v, T)$, then there was a successful call to \fn{publish$_p$}$(v, T)$ with serialization point before time $t$.
\end{lemma}
\begin{proof}
S must have been modified by a successful CAS in line~\ref{hist-pub-later-sp}. Obviously, a successful invocation of \fn{publish$_p$}$(v, T)$ must have executed that CAS. Let's denote it by $\pi$.
This invocation sets $L[p]$ to $(v, T)$. What remains to be proven is that no later invocation of \fn{publish$_\cdot$} will set $L[p]$ to $(v, x)$ for any x.
Only invocations of \fn{publish$_p$} modify $L[p]$ and they can start only after $\pi$ has finished. Together with a simple observation that $S.v$ is nondecreasing in time this implies
that any later invocation of \fn{publish} that sets $L[p]$ must have been called with a strictly greater $v$.
\end{proof}

\begin{observation}
\label{obs-hist-is-correct}
If $H[i] = (v, T)$ at some point then a successful call to \fn{publish$_\cdot$}$(v, T)$ has had its serialization point earlier.
\end{observation}

\begin{lemma}\label{lem-hist-help-comes}
Let $\pi$ be a call to \fn{publish$_p$}. Let $\pi_q$ be a successful call to \fn{publish$_q$} and $\pi_p$ a call to \fn{publish$_p$}, both with serialization points after serialization point of $\pi$. Then there is a call to \fn{help} that starts after the serialization point of $\pi$ and ends before the one of $\pi_q$ and 
before the execution of line~\ref{hist-pub-latest} in $\pi_p$.
\end{lemma}
\begin{proof}
The serialization point of $\pi$ must happen before $\pi_q$ executes line~\ref{hist-pub-earlier-sp} -- if S has changed between line~\ref{hist-pub-earlier-sp} and line~\ref{hist-pub-later-sp} of $\pi_q$, $\pi_q$ would fail.
Thus the call to \fn{help} from $\pi_q$ will start after the serialization point of $\pi$ and will finish before the serialization point of $\pi_q$. The call to \fn{help}
from $\pi_p$ will start after serialization point of $\pi$ ($\pi_p$ may only start after $\pi$ has finished) and will finish before line~\ref{hist-pub-latest} of $\pi_p$.
Of these two calls to \fn{help}, the one that finishes earlier satisfies both conditions.
\end{proof}

\begin{lemma}
\label{lem-hist-is-complete}
If a call to \fn{help}() has executed fully (ie. started and finished) after the serialization point of a successful \fn{publish$_\cdot$}$(v, \cdot)$, then $H[v\mod N].v \geq v$.
\end{lemma}
\begin{proof}
Let us first note that $H[v\mod N].v$ is nondecreasing.
It thus suffices to prove that the condition is met at some point before the end of the call to \fn{help}. We will do so by induction on the time at which the invocation of \fn{publish} serializes.
Consider a call to \fn{publish$_p$}$(v, \cdot)$ that serializes at time $t$, and a call to \fn{help} that starts after $t$. By Lemma~\ref{lem-hist-help-comes} there is a call to \fn{help} that finishes before the next serialization point of a successful \fn{publish$_\cdot$} (thus before $S$
is modified) and before $L[p]$ is modified. Without loss of generality we can assume that the call to \fn{help} we are considering satisfies these conditions. By inductive hypothesis and Lemma~\ref{lem-hist-help-comes}, since a successful call to \fn{publish$_\cdot$}$(v-N+1, \cdot)$ has occured before time $t$, $H[v \mod N].first \geq v-N$ at time $t$. Taking into account that $H[i].v \mod N = i$, at any later point in time one of $H[v \mod N].v = v-N$ and $H[v \mod N] \geq v$ will hold.
Thus the CAS in line~\ref{hist-help-cas} either succeeds, which causes the lemma's conclusion to start being satisfied, or fails, which means that the conclusion was already satisfied.
\end{proof} 

\begin{theorem}
The operations in the implementation of the history object serialize to the single-process object with serialization points of \fn{publish$_\cdot$} as posited earlier.
\end{theorem}
\begin{proof}
Note that $S.v$ is at all times equal to the first argument of the latest successful \fn{publish}.

Let us first consider a call to \fn{get-current} that returns $r$. From Lemma~\ref{lem-hist-is-complete} we know that $r.v \geq v.v$.
By Observation~\ref{obs-hist-is-correct}, the successful \fn{publish$_\cdot$}$(r.v, \cdot)$ happened before the read from $H[v.v\mod N]$. If $r.v = v.v$, then this call to \fn{publish} was the most recent successful \fn{publish} at the time when $S$ was read, so we can serialize the operation there. Otherwise, it has happened (had its serialization point) after that read, so we can serialize the operation just after it.

Let us consider a calls to \fn{publish$_\cdot$}:
\begin{itemize}
\item If a call to \fn{publish$_\cdot$}$(v, \cdot)$ fails in line~\ref{hist-pub-earlier-exit}, then we can see that the most recently serialized (from the point of view of line~\ref{hist-pub-earlier-sp}) successful \fn{publish} published a version different than $v-1$.
\item If a call to \fn{publish$_\cdot$}$(v, \cdot)$ fails by failing the CAS in line~\ref{hist-pub-later-sp}, then a successful \fn{publish$_\cdot$} has occured after line~\ref{hist-pub-earlier-sp} had been executed, so the most recent successful \fn{publish$_\cdot$} at the time of \fn{CAS} has published a version greater than $v-1$.
\item If a call to \fn{publish$_\cdot$}$(v, \cdot)$ succeeds, then at the time of CAS in line~\ref{hist-pub-later-sp} the most recently published version is $v-1$.
	\end{itemize}
	This suffices to prove that the value returned by \fn{publish} is always correct with respect to the posited
	serialization order. What remains to consider are the calls to \fn{get}:
\begin{itemize}
	\item If a call to \fn{get}$(v)$ fails in line~\ref{hist-get-notyet}, then the most recent successful \fn{publish$_\cdot$} when line~\ref{hist-get-notyet-sp} executed had a smaller version number than requested, so we can serialize the call to \fn{get} at line~\ref{hist-get-notyet-sp}.
	\item If a call to \fn{get}$(v)$ fails in line~\ref{hist-get-tooold} because $h.v > v$, then (by Observation~\ref{obs-hist-is-correct}) a successful call to \fn{publish$_\cdot$}$(h.v, \cdot)$ has occured by that time and $h.v \geq v+N$, so if we serialize \fn{get} at that point, it should fail.
	\item If a call to \fn{get}$(v)$ fails in line~\ref{hist-get-tooold} because $h.v < v$, then (by Lemma~\ref{lem-hist-is-complete}) a call to \fn{publish$_\cdot$}$(v, \cdot)$ has not occured before the call to \fn{help} in line~X. We can thus serialize this \fn{get} just before the call to \fn{help} has started.
	\item If a call to \fn{get}$(v)$ succeeds, then by Observation~\ref{obs-hist-is-correct} it returns a value that was published by a successful call to \fn{publish$_\cdot$}$(v, \cdot)$. By Lemma~\ref{lem-hist-is-complete}, version $v+N$ was not published before the call to \fn{help} has started. Thus we can serialize this \fn{get} just after the successfull call to \fn{publish$_\cdot$}$(v, \cdot)$ has occured, or, if it has occured before \fn{get} started, just before the call to \fn{help} started.
\end{itemize}
\end{proof}

Obviously all operations run in $O(1)$ time. Memory complexity is $O(N+P)$, where $N$ and $P$ are the parameters defined at the beginning of this section.
We will now use this object in the main result of this work, an implementation of a write-and-$f$-array.
\chapter{The Write-and-f-array}

We will now present our main result -- an implementation of a write-and-f-array.
We will actually present a wait-free implementation of a slightly richer object -- the additional operations and return values are required for the recursive construction of the concurrent implementation.

A single-threaded specification of the structure is shown below. It uses a \fn{version} function, which is a black-box integer-valued function, subject to following conditions:
\begin{enumerate}
	\item Subsequent return values of \fn{version} are nondecreasing.
	\item If a call to \fn{version}(\textbf{false}) is followed by a call to \fn{version}(\textbf{true}), the second call must return a strictly greater integer than the first one.
\end{enumerate}

\begin{algorithmic}[1]
	\State\var{$v[N]$}{\arrayspec{N}$\T$}
	\State\var{$\text{last\_update}$}{\arrayspec{N}\int}
	\State\var{$\text{last\_version}$}{\arrayspec{N}\int}
	\State\var{$\text{last\_value}$}{\arrayspec{N}\T}
	
	\Function{write-and-f$_i$}{$T$}
		\State $v[i] \gets T$
		\State $r \gets f(v[0], v[1],\ldots, v[N-1])$
		\State $\text{last\_update}[i] \gets \text{last\_update}[i] + 1$
		\State $\text{last\_version}[i] \gets \fn{version}(\textbf{true})$
		\State $\text{last\_value}[i] \gets r$
		\State \Return $\text{last\_update}[i], \text{last\_version}[i], \text{last\_value}[i]$
	\EndFunction
	\Function{get-last}{$i$}
		\State \Return $\text{last\_update}[i], \text{last\_version}[i], \text{last\_value}[i]$
	\EndFunction
	\Function{read}{}
		\State \Return $\fn{version}(\textbf{false}), f(v[0], v[1],\ldots, v[N-1])$
	\EndFunction
\end{algorithmic}

One can observe that the version numbers group the calls to \fn{write-and-f} into groups of consecutive calls with no 
intervening calls to \fn{read}.

The concurrent implementation will be restricted by disallowing concurrent calls to \fn{write-and-f$_i$} for the same $i$.
It is conceptually very similar to Jayanti's tree-based f-array implementation. It uses a binary tree structure, with each array element assigned to a leaf and intervals of array elements assigned to internal nodes.
We will construct it recursively. The implementation for $N = 1$ is presented below. If we note that no concurrent calls to \fn{write-and-f} may be made in it, we can easily
see that it is indeed correct and that all operations take constant time.

\begin{algorithmic}[1]
	\State\var{$S$}{$\textbf{pair}<v: \int, T: \T>$} 
	\Function{write-and-f$_0$}{T'}
		\State $v, T \gets S.\fn{rd}()$
		\State $S.\fn{wr}((v + 1, T'))$
		\State \Return $v + 1, v + 1, T'$
	\EndFunction
	\Function{get-last}{$i$}
		\State $v, T \gets S.\fn{rd}()$
		\State \Return $v, v, T$
	\EndFunction
	\Function{read}{}
		\State \Return $S.\fn{rd}()$
	\EndFunction
\end{algorithmic}

The implementation for $N > 1$ is presented below (interspersed with comments):
\begin{algorithmic}[1]
	\Struct{node-value}
		\State \var{$T$}{\arrayspec{2}\T} \Comment{Child values}
		\State \var{$v$}{\arrayspec{2}\int} \Comment{Child versions}
	\EndStruct
	\Struct{last-value}
		\State \var{$n$}{\int}
		\State \var{$v$}{\int}
		\State \var{$T$}{\T}
	\EndStruct
	\State \var{$C$}{\arrayspec{2}\textsc{write-and-f-array}}  \Comment{of sizes $\lceil{}N/2\rceil{}$ and $\lfloor{}N/2\rfloor{}$, resp.}
	\State \var{$H$}{\textbf{history object}<\textsc{node-value}>} \Comment{of size $N+1$, for $N$ concurrent updates}
	\State \var{$L$}{\arrayspec{N}\textsc{last-value}} 
	\algstore{impl}
\end{algorithmic}

C[0] and C[1] are the subobjects from the recursive construction -- they are of size $\lceil{}N/2\rceil{}$ and $\lfloor{}N/2\rfloor{}$, respectively. Array elements of the enclosing
structure are mapped bijectively to consecutive array elements of these two subobjects. The mapping is defined by following functions (the element $i$ is mapped to element $child\_id(i)$ in C[$side(i)$]):
\begin{equation*}
side(i) = \begin{cases} 0&i\in\{0,\ldots,\lceil{}N/2\rceil{}-1\}\\1&i\in\{\lceil{}N/2\rceil{},\ldots,N-1\}\end{cases}
\end{equation*}
\begin{equation*}
child\_id(i) = \begin{cases} i&i\in\{0,\ldots,\lceil{}N/2\rceil{}-1\}\\i-\lceil{}N/2\rceil{}&i\in\{\lceil{}N/2\rceil{},\ldots,N-1\}\end{cases}
\end{equation*}
History object H is used to store the object's current value (the value that would be returned by \fn{read}) and past values. The version numbers exported by the history object correspond to the values returned by \fn{version}
in the specification. The elements of H aren't just values; they instead contain the values of both children together with their versions.
The array L is used to store the values \fn{get-last} should return, but the values there might be stale (we will prove bounds on their staleness later on).

\begin{algorithmic}[1]
	\algrestore{impl}
	\Function{read}{}
		\State $v, h \gets H.\fn{get-current}()$ \label{get-getownver}
		\State \Return $v, f(h.T[0], h.T[1])$
	\EndFunction
	\Function{update$_i$}{}
		\State $v, h \gets H.\fn{get-current}()$ \label{upd-getown}
		\State $h'.v[0], h'.T[0] \gets C[0].\fn{read}()$ \label{upd-getch-0}
		\State $h'.v[1], h'.T[1] \gets C[1].\fn{read}()$ \label{upd-getch-1}
		\State $\fn{help}(v \mod N)$
		\State \Return $H.\fn{publish$_i$}(v+1, h')$ \label{upd-publish}
	\EndFunction
	\Function{write-and-f$_i$}{T}
		\State $C[side(i)].\fn{write-and-f$_{child\_id(i)}$}(T)$ \label{upd-chupd}
		\If{\textbf{not} \fn{update$_i$}()} \label{updcall1}
		\State \fn{update$_i$}() \label{updcall2}
		\EndIf
		\State \Return $\fn{get-last}(i)$
	\EndFunction
	\algstore{impl}
\end{algorithmic}

The implementation of \fn{read} and \fn{write-and-f} strongly resemble the f-array. \fn{write-and-f} uses a helper function \fn{update}. The intuition behind \fn{update} is that it ``pushes'' new values from $C[0]$ and $C[1]$ to $H$.
We will show that it suffices to attempt to call \fn{update} twice to accomplish that.

\begin{algorithmic}[1]
	\algrestore{impl}
	\Function{help}{x}
		\State $l_c \gets C[side(x)].\fn{get-last}(child\_id(x))$ \label{get-last-recurse}
		\State $v_{current} \gets H.\fn{get-current}().v$
		\State \begin{varwidth}[t]{\linewidth} binary search $\{v_{current}-(N+1),\ldots, v_{current}\}$ for first $v$ such that: \label{help-binsch}\par
		\hskip\algorithmicindent $H.\fn{get}(v) \neq \none \textbf{ and } H.\fn{get}(v).v[side(x)] \geq l_c.v$ \end{varwidth}
		\If{no such $v$ exists}
			\State \Return
		\EndIf
		\State $h_{old} \gets H.\fn{get}(v-1)$ \label{help-getprev}
		\State $h_{new} \gets H.\fn{get}(v)$
		\If{$h_{old} = \none \textbf{ or } h_{new} = \none$}
			\State \Return
		\EndIf
		\If{$side(x) = 0$}
			\State $T \gets f(l_c.T, h_{old}.T[1])$
		\Else
			\State $T \gets f(h_{new}.T[0], l_c.T)$
		\EndIf
		\State $l \gets L[x].\fn{rd}$
		\If{$l.n < l_C.n$} \label{upd-last-alreadydone}
			\State $L[x].\text{CAS}(l, (l_C.n, v, T))$
		\EndIf
	\EndFunction
	\Function{get-last}{x}
		\State $\fn{help}(x)$
		\State \Return $L[x].\fn{rd}()$
	\EndFunction
\end{algorithmic}

The use of binary search in line~\ref{help-binsch} warrants explanation. The predicate employed can obviously change its value in time. Thus, the binary search will return a $v$ such that the predicate held at one point in time
for $v$ and didn't at another point for $v-1$. We will see that this is sufficient. Indeed, the result of the binary search will be important only in the cases when the value of the predicate doesn't change during
the search.

We will first prove two simple facts about \fn{update}:

\begin{lemma}
\label{lem-versions-increase}
For any two subsequent successful calls to $H.\fn{publish}$ with values $h_1$ and $h_2$, $h_1.v[i] \leq h_2.v[i]$ for both $i$.
\end{lemma}
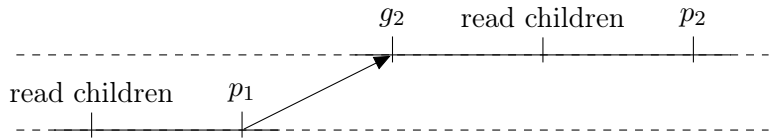
\begin{figure}[!h]
\begin{tikzpicture}
	\atomic{prev_getch}{read children}{1}{1}
	\atomic{prev_pub}{$p_1$}{3}{1}

	\atomic{get_ver}{$g_2$}{5}{0}
	\atomic{getch}{read children}{7}{0}
	\atomic{publish}{$p_2$}{9}{0}

	\draw[dashed] (0,-1) to (10,-1);
	\draw[dashed] (0,0) to (10,0);
	\draw (0.5,-1) to (3.5,-1);
	\draw (4.5,0) to (9.5,0);

	\draw[arrows={-triangle 45}] (prev_pub) to (get_ver);
\end{tikzpicture}
\caption{Diagram for proof of Lemma~\ref{lem-versions-increase} ($p_i$ -- call to \fn{publish}; $g_i$ -- call to \fn{get-current} in line~\ref{upd-getown} that corresponds to $p_i$).}
\label{fig-versions-increase}
\end{figure}
\begin{proof}
Assume otherwise. Consider the first pair $p_1, p_2$ of consecutive calls to \fn{publish} that contradicts the Lemma (see Figure~\ref{fig-versions-increase}). The only place \fn{publish} is called is line \ref{upd-publish}
in \fn{update}.
$p_2$ had its most recent call to \fn{get-current} (denoted by $g_2$) in line~\ref{upd-getown} sometime earlier. No successful publish could have occured between $g_2$ and $p_2$; otherwise $p_2$
would have failed.
Thus, $p_1$ must have occured before $g_2$ and so its execution of lines \ref{upd-getch-0} and \ref{upd-getch-1} (denoted by ``read children'' on the diagram) must have occured strictly earlier than
the same for the second call. This together with monotonicity of versions in children delivers the contradiction.
\end{proof}

\begin{lemma}
	For every execution of \fn{write-and-f}: during execution of lines \ref{updcall1} and \ref{updcall2} at least one successful call to \fn{update} occurs.
\end{lemma}
\begin{figure}
\begin{tikzpicture}
	\atomic{our_gv1}{$g_1$}{5.5}{0}
	\atomic{bad_pub1}{$p^S_3$}{6.5}{1}
	\atomic{our_pub1}{$p^F_1$}{7.5}{0}
	\draw[arrows={-triangle 45}] (our_gv1) -- (bad_pub1);
	\draw[arrows={-triangle 45}] (bad_pub1) -- (our_pub1);
	\atomic{bad_gv2}{$g_4$}{7}{2}
	\draw[arrows={-triangle 45}] (bad_pub1) -- (bad_gv2);
	\atomic{our_gv2}{$g_2$}{9.5}{0}
	\atomic{bad_pub2}{$p^S_4$}{11}{2}
	\atomic{our_pub2}{$p^F_2$}{12}{0}
	\draw[arrows={-triangle 45}] (our_gv2) -- (bad_pub2);
	\draw[arrows={-triangle 45}] (bad_pub2) -- (our_pub2);

	\draw[dashed] (2.5,0) -- (13,0);
	\draw (3,0) -- (12.5,0);

	\draw[dashed] (5.5,-1) -- (7.5,-1);
	\draw (6,-1) -- (7,-1);

	\draw[dashed] (6,-2) -- (12,-2);
	\draw (6.5,-2) -- (11.5,-2);
\end{tikzpicture}
\caption{Consequences of two failed \fn{update} calls ($g_i$ -- call to \fn{get-current} in line~\ref{upd-getown}, $p^S_{i}$ -- successful call to \fn{publish} in line~\ref{upd-publish}, $p^F_{i}$ -- failed call to \fn{publish} in line~\ref{upd-publish}).}
\label{fig-only-two-pubs}
\end{figure}
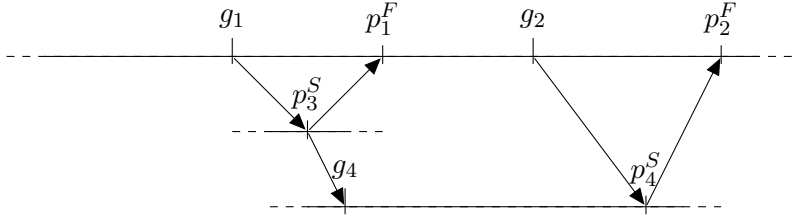

\begin{proof}
	If one of the calls to \fn{update} from the lines in question suceeds, then the Lemma trivially holds. Thus we assume that they have both failed. The situation is depicted
	in Figure~\ref{fig-only-two-pubs}. $g_1, p^F_1, g_2 \text{ and } p^F_2$ correspond to these two failed calls to \fn{update}. In order for $p^F_1$ to fail, a \fn{publish} must have
	succeeded between $g_1$ and $p^F_1$ -- let's call it $p^S_3$. Similarly, a \fn{publish} must have suceeded between $g_2$ and $p^F_2$. Alas, this \fn{publish} (denoted $p^S_4$) must have had
	its corresponding call to \fn{get-current} (denoted $g_4$) after $p^S_3$. The call to \fn{update} that called $g_4$ and $p^S_4$ satisfies the Lemma's conclusion.
\end{proof}

Let us consider an execution of \fn{write-and-f$_i$}. Let $(n_c, v_c, T_c)$ be the return value of the corresponding call to $C[side(i)].\fn{write-and-f$_{child\_id(i)}$}$. Then by the time the call to
\fn{write-and-f$_i$} finishes, the value of the most recently published node has $v[side(i)] \geq v_c$. We will consider the first call to \fn{publish} that publishes such a node value as the serialization
point of the \fn{write-and-f}. Obviously, many calls to \fn{write-and-f} can then serialize at the same instant in time. We will order them by taking first these with $side(i) = 0$ in the order in which
their corresponding calls to \fn{write-and-f} happened in $C[0]$ and then those with $side(i) = 1$ in the corresponding order. The choice of order on elements of $C$ is arbitrary, albeit it is reflected
in the computation of $T$ in \fn{help}.

Before we prove that the implementation is correct, we need the following two lemmas about the array $L$.

\begin{lemma}
	If \fn{help} is called after the serialization point of the $n^{\text{th}}$ call to \fn{write-and-f$_i$}, then $L[i].n \geq n$ after the call to \fn{help} completes. \label{last-is-complete}
\end{lemma}
\begin{proof}
	We will prove this Lemma by induction on $n$. The proof for the base case will be very similar to the induction step, so we present both of them simultaneously.
	Let $p(v)$ be the time the successfull call to \fn{publish}$(v, \cdot)$ occured. Consider an \fn{update$_i$} that was serialized at $p(v)$.
	For every $u \equiv i \left(\mbox{mod } N\right)$, a call to $\fn{help}(i)$ starts and finishes between $p(u-1)$ and $p(u)$. Thus, there is such a call that starts and finishes between $p(v)$ and $p(v+N)$.
	Without loss of generality we can assume that the call to \fn{help} from the Lemma's hypothesis executes in this time interval. In that case, the binary search from line~\ref{help-binsch} will not see any
	$\fn{get}(x)$ for $x \geq v-1$ returning \none and will	return version $v$. For the same reason, both \fn{get}s will succeed.

	We want to prove that $L[i].n \geq n-1$ from $p(v)$ on. For the base case this holds, because the initial values satisfy this requirement. For the inductive step, we have to use the hypothesis:
	at $p(v)$, the $n-1^{\text{th}}$ call to \fn{write-and-f$_i$} has completed, so a call to $\fn{help}(i)$ has completed between the serialization point of that call to \fn{write-and-f$_i$}
	and $p(v)$. Thus $L[i].n \geq n-1$ from $p(v)$ on. If the condition in line~\ref{upd-last-alreadydone} is not satisfied or the CAS in the next line fails, the Lemma hold.
	Otherwise, by the time the CAS succeeds the Lemma will obviously hold.
	
\end{proof}

\begin{lemma}
	If a call to $\fn{help}(i)$ sets $L[i]$ to $(n, v, T)$, then by the time the call has finished: \label{last-is-correct}
\begin{enumerate}
	\item \fn{write-and-f$_i$} has been called at least $n$ times.
	\item Let $(n_c, v_c, T_c)$ be the return value of the $n^{\text{th}}$ call to $C[side(i)].\fn{write-and-f}_{child\_id(p)}$. Let $v'$ and $h_{new}$ be the version and node value published at the serialization point of the $n^{\text{th}}$
		\fn{write-and-f$_i$} and $h_{old}$ be the node value published with version $v'-1$. Then $v = v'$ and:
		\begin{itemize}
			\item if $side(p) = 0$ then $T = f(T_c, h_{old}.T[1])$,
			\item if $side(p) = 1$ then $T = f(h_{new}.T[0], T_c)$.
		\end{itemize}
\end{enumerate}
\end{lemma}
\begin{proof}
	For the first part, we need to note that by the time line~\ref{get-last-recurse} executes, the $n^{\text{th}}$ \fn{write-and-f$_{child\_id(i)}$} has already executed in the appropriate child.
	If the binary search fails or returns version different from $v'$, \fn{help} will exit early. Indeed, if the binary search returns a different $v$, $\fn{get}(v-1)$ will fail.
	This immediately proves the first part of the Lemma and shows that $h_{old}$ and $h_{new}$ in the code have the same values as $h_{old}$ and $h_{new}$ in the Lemma, which proves the second part.
\end{proof}

\begin{theorem}
	The operations in multi-process write-and-f-array serialize to the single-process specification with serialization point of \fn{write-and-f} as posited earlier.
\end{theorem}
\begin{proof}
	Let us choose the call to \fn{get-current} as the serialization point of \fn{read}. 
	We will first prove that these are correct serialization points for \fn{write-and-f} and \fn{read}.
	Consider the $n^{\text{th}}$ call to \fn{write-and-f$_i$} and let $(n', v', T')$ be its return value. As this triple was read from $L[i]$, we can use Lemmas~\ref{last-is-complete} and~\ref{last-is-correct} to show that
	$n' = n$ (note that executions of \fn{write-and-f$_i$} for the same $i$ are disjoint in time). From Lemma~\ref{last-is-correct} we get that $v'$ is equal to the
	version published at the serialization point of the call to \fn{write-and-f}. This obviously implies that \fn{version} is nondecreasing for subsequent calls to \fn{write-and-f} and \fn{read}.
	As only the calls to \fn{write-and-f} that serialize on a single \fn{publish} return equal $v$, no \fn{read} call can serialize in between. This proves both properties required from \fn{version}.

	We still need to prove
		that $T'$ and the values returned by \fn{read} are consistent with the posited serialization order.
		Consider a set of \fn{write-and-f} calls that are serialized together. Let $h_{new}$ and $h_{old}$ be the history elements published, respectively, at the serialization point and most recently before it. Let $a_i$ be the sequence of return values
	of $C[0].\fn{write-and-f}$ calls corresponding to \fn{write-and-f} calls in question, in serialization order. Let $b_i$ be a similar sequence for $C[1]$. By Lemma~\ref{lem-versions-increase} these are exactly the calls to
	$C[j].\fn{write-and-f}$ serialized with versions in $\left(h_{old}.v[j], h_{new}.v[j]\right]$ for $j \in \{0, 1\}$. By Lemma~\ref{last-is-correct} the sequence returned
by the \fn{write-and-f} calls is $f(a_0, b_0), $ $f(a_1, b_0),$ $\ldots, $ $f(h_{new}.T[0], b_0), $ $f(h_{new}.T[0], b_1), \ldots$ in the order of serialization. 
Additionally, $f(h_{new}.T[0], h_{new}.T[1])$ is the value that would be returned by any \fn{read} calls until the next serialization point of \fn{write-and-f}.
This suffices to show that the return values of consecutive calls to \fn{write-and-f} are actually the results of applying the updates.
This proves that values returned by \fn{get} and \fn{update} are correct wrt. posited serialization order.
	
	We still need to show that we can correctly serialize the calls to \fn{get-last}. Lemma~\ref{last-is-correct} implies that two calls to $\fn{get-last}(i)$ will never return different triples with equal $n$. As \fn{write-and-f} returns
	the result of a call to \fn{get-last}, we just need to show that \fn{get-last}$(i)$ can be serialized so that the $n$ in its return value is the number of previously serialized calls to \fn{write-and-f$_i$}. From
	Lemma~\ref{last-is-complete} we get that $n$ is at least the number of calls to \fn{write-and-f$_i$} that serialized before \fn{get-last} started. From Lemma~\ref{last-is-correct} we know that $n$ was at most the number of such calls
	that serialized before \fn{get-last} finished. Thus, there is a point in the execution interval of \fn{get-last} when $n$ is equal to the number of such calls that have already serialized. We choose any such point in time to be the serialization
	point of the call to \fn{get-last}.
\end{proof}

The structure uses $O(N\log{}N)$ memory. A \fn{get-last} operation takes $O(\log^2 N)$ time, a \fn{read} operation takes $O(1)$ time and a \fn{write-and-f} operation takes $O(\log^3 N)$ time.

We can construct a fetch-and-add object for $P$ threads by using a write-and-f-array of size $P$ with addition as the operation $f$. Every thread is assigned an element in the array, and modifies only that element. This gives us
a fetch-and-add object for $P$ threads with $O(\log^3 P)$ step complexity and $O(P\log P)$ memory complexity. Additionally, the object implements a method that retrieves the
current value in $O(1)$ time.

\chapter{Implementation}
A fetch-and-add object was implemented using the above-mentioned construction from the write-and-f object. The implementation is written in C++ and uses the C++11 standard library support for atomic operations. It is published on Github \footnote{The version current as of the time this work is written is available at: \verb+https://github.com/robryk/parsum/tree/d3433f7f7b137b52a73cfb244d46081528c696c3+ }.

Direct implementation of previously described algorithm would require CAS objects of size larger than the
64-bit Intel processors support. However, the way most of these objects are accessed in the implementation makes
it possible to split them into multiple CAS objects. The only CAS objects that don't afford this transformation
are the version number holders from the history object. We use 64-bit versions and thus even these objects are
no larger than 128 bits. Thus, our implementation can run on 64-bit Intel architecture processors, but not on 32-bit ones.

Unfortunately, the C++ standard library support for atomic operations doesn't allow atomic reads of a part of a larger atomic variable (specifically, reads of 64-bit halves of a 128-bit variable).
As 128-bit atomic reads on amd64 are very costly
(they use a 128-bit wide CAS), limiting their number was very important for the efficiency's sake. Thus, the implementation makes an unwarranted assumption that an atomic variable
has the same memory layout as a normal, non-atomic variable of the same size. This assumption holds for gcc 4.6 on amd64.

The correctness of the implementation was tested experimentally both on real hardware and by using the Relacy Race Detector\cite{rrd}. Relacy Race Detector is a framework for testing multi-threaded programs in C++11.
It substitues its own
implementation of synchronization primitives and atomic variables and runs user-supplied tests multiple times, with various interleavings of threads.
It can detect data races on non-atomic variables, deadlock conditions, and failed user-supplied invariant and assertion
checks. One notable feature is the support for atomic operations with reduced consistency -- Relacy can simulate acquire/release semantics, as specified in the C++11 atomic variables library.
Our fetch-and-add implementation can be compiled both to run on bare hardware and to run in Relacy. The use of Relacy not only provided confidence about correctness of the implementation, but also allowed us to downgrade
consistency guarantees of some writes.

The implementation was benchmarked on a machine with 4 12-core AMD Opteron 6174 processors.
The benchmark created a fetch-and-add object and started a preset number of threads. Each thread incremented the counter in a loop, counting its operations locally. After 20 seconds elapsed, all threads were signalled
to stop and their operation counters were summed. For comparison, the same benchmark was run with the fetch-and-add object implemented by a simple read-modify-write loop. All benchmarks were run when the machine was
minimally loaded (had 5 minute load average smaller than 0.5 at the start of the benchmark).

The results of the benchmark are shown in Figure~\ref{fig-meas}. The figure contains a plot of average time it takes to complete one operation as a function of the number of concurrently executing threads. 
The horizontal scale of the figure is proportional to the square of the logarithm of the number of threads. The reason for this is that the optimistic time complexity
of our fetch-and-add object is $\Theta(\log^2 P)$. We can see that the graph for our implementation approximates a straight line up to about 30 threads, where it starts to diverge upward. This peculiar divergence disappears
if we assign one core to each thread and force it to run only there (by setting CPU affinity). Other affinity settings yielded graphs similar to one of these two. We were unable to explain this behaviour. It might
be worth noting that the anomaly in the no-affinity case happens when we start using all 4 physical CPUs.

It can be easily seen from the graph that our fetch-and-add object is far from practical. It is about 100 times slower than the naive lock-free implementation for small numbers of threads and about 10 times slower
for larger numbers of threads. We believe that this implementation is suboptimal, although achieving similar performance as the naive implementation seems impossible. One change that could improve the performance
of this implementation is changing the arity of the tree. The construction can easily be adapted to trees with larger arity and this might decrease the number of expensive CAS operations at the expense of the number
of atomic reads, which are comparatively cheap on Intel processors.

\begin{figure}
\includegraphics{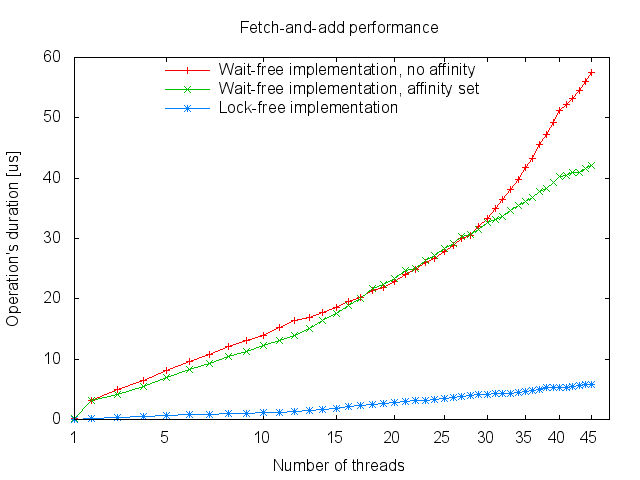}
\caption{Results of the benchmark}
\label{fig-meas}
\end{figure}

\chapter{Extensions and Open Problems}

Our implementation of write-and-f-array can be modified by splitting the work done by \fn{help} into $O(\log N)$ separate pieces (updates on consecutive levels of the tree) and running only a single such piece in \fn{update}.
If we then enlarge the history size to $O(N\log N)$, Lemma~\ref{last-is-complete} still holds. This modification increases the memory complexity to $O(N\log^2N)$ and decreases the step complexity of
\fn{write-and-f} to $O(\log^2P)$.

The structure can be straightforwardly modified to use LL/CS instead of CAS, albeit it then requires the ability to have two outstanding LLs at the same time. The modified version can be further modified to use
bounded version numbers, from a range of size $O(N)$. Unfortunately, popular architectures that implement LL/SC (such as PowerPC) allow only one outstanding LL at a time.

The effect of the arity of the tree on the performance
seems to be worth investigating.

We pose also two purely theoretical problems:

\begin{itemize}
	\item Our implementation of write-and-f-array is single-writer, so it is natural to consider the multi-writer write-and-f-array. Can we construct a multi-writer write-and-f-array in subquadratic memory with similar step complexities of the operations?

	\item The memory consumption of write-and-f-array implemented with atomic registers and CAS or LL/CS objects is bounded from below by a linear function of the number of processes that can execute \fn{write-and-f} concurrently\cite{lowerbound}, so it must be $\Omega(N)$. Can we achieve lower memory and/or step complexities?
\end{itemize}

\begingroup\let\clearpage\relax\chapter{Acknowledgements}\endgroup
The author would like to thank his advisor, dr. Grzegorz Herman for many fruitful discussions and help while preparing this work. He would also like to thank Szymon Acedański for his help in improving the presentation of this work.

This work was supported by the Polish Ministry of Science and Higher Education program ``Diamentowy Grant''.


\newpage

\begin{thebibliography}{9}
	\bibitem{TheArt} Herlihy M., Shavit N. ``The Art of Multiprocessor Programming'', Morgan Kaufmann Publishers Inc., 2008
	\bibitem{cas-univ} Herlihy M. ``Wait-free synchronization'', ACM Trans. Program. Lang. Syst. 13, 1, pp. 124--149 (1991)
	\bibitem{snapshot-impl} Riany Y., Shavit N., Touitou D. ``Towards a practical snapshot algorithm'', Theoretical Computer Science 269, pp. 163--201 (2001)
	\bibitem{snapshot-idea} Afek Y., Attiya H., Dolev D., Gafni E., Merritt M., Shavit N. ``Atomic snapshots of shared memory'', Journal of the ACM 40, pp. 873--890 (1993)
	\bibitem{write-and-snap} Afek Y., Weisberger E. ``The instancy of snapshots and commuting objects'', Journal of Algorithms 30.1, pp. 68--105 (1999)
	\bibitem{closed-object} Chandra T. D., Jayanti P., Tan K. ``A polylog time wait-free construction for closed objects'', Proceedings the 17th PODC, pp. 287--296 (1998)
	\bibitem{f-array} Jayanti P. ``f-arrays: Implementation and applications'', Proceedings of the 21st PODC, pp. 270--279 (2002)
	\bibitem{ellen-fai} Ellen F., Ramachandran V., Woelfel P. ``Efficient Fetch-And-Increment'', Lecture Notes in Computer Science 7611, pp. 16--30 (2012)
	\bibitem{lowerbound} Fich F. E., Hendler D., Shavit N. ``Linear lower bounds on real-world implementations of concurrent objects'', Foundations of Computer Science, pp. 165--173 (2005)
	\bibitem{faa-sync} Freudenthal, E., Gottlieb, A. ``Process coordination with fetch-and-increment'', Proceedings of ASPLOS-IV, pp. 260–-268 (1991)
	\bibitem{faa-queue} Goodman, J., Vernon, M., Woest, P. ``Efficent synchronization primitives for large-scale cache-coherent multiprocessors'', Proceedings of ASPLOS-III, pp. 64-–75 (1989)
	\bibitem{rrd} Vyukov D., Relacy Race Detector, \verb+http://www.1024cores.net/home/relacy-race-detector+
\end{thebibliography}
\end{document}